\documentclass[11pt,letterpaper]{amsart}

\usepackage{amssymb,latexsym,amsmath,extarrows, mathrsfs, amsthm,color,amsrefs}
\usepackage{graphicx}

\usepackage[margin=1in, centering]{geometry}
\usepackage[bookmarksnumbered, colorlinks, plainpages]{hyperref}
\usepackage{hyperref} %,hypertexnames=false,colorlinks,[pagebackref]
\hypersetup{
    colorlinks=true,       % false: boxed links; true: colored links
    linkcolor=blue,          % color of internal links
    citecolor=magenta,        % color of links to bibliography
    filecolor=magenta,      % color of file links
    urlcolor=cyan           % color of external links
}

\usepackage{mathtools}
\mathtoolsset{showonlyrefs,showmanualtags}

\newtheorem{theorem}{Theorem}[section]

\newtheorem{lemma}[theorem]{Lemma}
\newtheorem{proposition}[theorem]{Proposition}
\newtheorem{remark}{Remark}[section]

\newtheorem{definition}{Definition}

\newcommand{\E}{\mathbb E}

\newcommand{\C}{\mathbb C}
\newcommand{\Z}{\mathbb Z}

\newcommand{\N}{\mathbb N}
\newcommand{\R}{\mathbb R}

\newcommand{\la}{\left\langle}
\newcommand{\ra}{\right\rangle}

\title[]{Multipoint}

\begin{document}

\date{\today}

\title[Multi-point correlation in $\mathbb{Z}^d$]{Decay of multi-point correlation functions in $\mathbb{Z}^d$}
\begin{abstract} 
We prove multi-point correlation bounds in $\mathbb{Z}^d$ for arbitrary $d\geq 1$ with symmetrized distances, answering open questions proposed by Sims-Warzel \cite{SW} and Aza-Bru-Siqueira Pedra \cite{ABP}.
As applications, we prove multi-point correlation bounds for the Ising model on $\mathbb{Z}^d$, and multi-point dynamical localization in expectation for uniformly localized disordered systems, which provides the first examples of this conjectured phenomenon by Bravyi-König \cite{BK}. 
\end{abstract}

\author{Rui Han}
\author{Fan Yang}
\maketitle

\section{Introduction}
Correlation functions arise from, among others, the study of quantum field theory and statistical mechanics, and are important measurements of the order of a system.
More specifically, correlation functions quantify how microscopic variables co-vary with one another on average across space and time.
For various systems, there have been intensive studies on the decay of the two-point correlation functions.
The goal of this paper is to use the two-point decay as input to study the decay of multi-point correlation functions for systems in arbitrary dimension.
For quasi-free states of a one-particle Hamiltonian, the time-dependent multi-point correlation function is the determinant of two-point functions.
For generalized Gaussian fields, the multi-point function can be written in terms of two-point functions via Wick rule.
One of our motivations comes from the Ising model on $\Z^d$.
Another motivation arises from the study of multi-point dynamical localization, which was proposed as a conjectured phenomenon in \cite{BK}, and was used as an assumption to study the growth of quantum memory.

We first present the focus of the paper in its simplest form.
Let $X=\{x_1,x_2,...,x_n\}\subset \Z^{d}$ and $Y=\{y_1,y_2,...,y_n\}\subset \Z^{d}$ be two sets of point configurations.
Throughout this paper, we will assume both $X$ and $Y$ consist of distinct points. For two points $x,y\in \Z^d$, let $|x-y|$ denote the Euclidean distance between them.
Let $H$ be a self-adjoint operator on $\ell^2(\Z^d)$.
One of the questions we consider, of its simplest form, can be formulated as:
suppose we have the decay of the two point function of the following form: 
\begin{align*}
|\la \delta_{x_j}, \rho(H) \delta_{y_k}\ra|\leq e^{-\eta |x_j-y_k|},
\end{align*}
for some positive constant $\eta$, some bounded function $\rho$ on $\R$ and any $1\leq j,k\leq n$.
What can be said about the decay of the multi-point correlation function, expressed through the determinant of the following matrix $M$?
\begin{align}\label{def:M_intro}
M=( \la \delta_{x_j}, \rho(H) \delta_{y_k} \ra)_{1\leq j,k\leq n}.
\end{align}
Such multi-point correlation functions have been studied in many papers, see e.g. \cite{BK,SW,ABP,ANSS}.
One of our goals of this paper is to study the decay of multi-point determinant in $\Z^d$, and in particular provide answers to open questions proposed in \cite{SW,ABP}.

In the rest of this section we present our main results. The main applications to the Ising model and dynamical localization are presented in Sections 2 and 3.
Before stating our results and related literature, we first introduce various distances between two sets of point configurations in $\Z^d$.

Our first result, Theorem \ref{thm:main2}, concerns the following symmetrized maximal distance:
\begin{align}\label{def:md}
D_m(X,Y):=\min_{\pi \in S_n} \max_{j=1}^n |x_j-y_{\pi(j)}|,
\end{align}
where $S_n$ is the symmetric group.

Our second result, Theorem \ref{thm:determinant}, concerns the symmetrized sum distance $D_s(X,Y)$ defined as:
\begin{align}\label{def:sd}
D_s(X,Y):=\min_{\pi\in S_n} \sum_{j=1}^n |x_j-y_{\pi(j)}|.
\end{align}

Another distance that we do not use in the present paper, but was used in the literature, is the Hausdorff distance:
\begin{align}
D_H(X,Y)=\max(\max_{x_j\in X}\mathrm{dist}(x_j, Y),\, \max_{y_k\in Y}\mathrm{dist}(y_k,X)),
\end{align}
where $\mathrm{dist}(w, V):=\min_{v\in V} |w-v|$ is the distance between a point $w$ and a set $V$. 
This distance was used in \cite{AW,ABP}. In particular \cite{ABP} extends some results of \cite{SW} to higher dimensions, in the $D_H$ distance setting.

It is clear that 
\begin{align*}
D_H(X,Y)\leq D_m(X,Y)\leq D_s(X,Y)\leq n D_m(X,Y).
\end{align*}
Hence the two distances $D_m$ and $D_s$ that we consider in this paper are stronger than $D_H$, see also some comments in \cite{AW}.

For $d=1$, Sims and Warzel \cite{SW} proved that $|\det M|$, for $M$ as in \eqref{def:M_intro}, decays with respect to the $D_m$ distance in both the deterministic case (when each entry of $M$ is a fixed scalar) and the disordered case (when each entry is a function and the result is in expectation). These results were further extended to arbitrary dimension, but with the weaker Hausdorff distance in \cite{ABP}.
The $D_m$ and $D_s$ distances remained a challenge in higher dimensions.
Sims and Warzel's proof \cite{SW}, and also a recent proof of dynamical localization for multi-particle systems in $d=1$ with respect to the $D_s$ distance \cite{BM}, both rely heavily on the fact that the points in $\Z^1$ are ordered, which can not be generalized to higher dimensions.
In higher dimensions, an additional difficulty lies in the complicated geometry involved.
It was proposed as an open question in \cite{SW,ABP} that whether there is a decay in arbitrary dimension with respect to the $D_m$ distance and even the stronger $D_s$ distance.

One of our main achievements in this paper is to give an positive answer to this question by developing new approaches in higher dimensions.
Our first main result proves a higher dimensional version of Theorems 1.1 and 1.2 of \cite{SW} with respect to the $D_m$ distance, in both deterministic and disordered cases, hence providing a direct answer to the open question mentioned above.
We formulate Theorem \ref{thm:main2} in the disordered case, but it clearly also holds in the deterministic case (when taking the $m^{\omega}_{jk}(t)$ below to be constant in $\omega$ and $t$).
For a matrix $M\in \C^{\ell\times m}$, the matrix norm $\|M\|$ is defined as its operator norm from $\ell^2(\C^m)$ to $\ell^2(\C^{\ell})$.

\begin{theorem}\label{thm:main2}
Let $K: [0,\infty)\to [0,\infty)$ be a non-negative valued function.
Suppose a matrix $M^{\omega}(t)=(m_{jk}^{\omega}(t))_{1\leq j,k\leq n}$ is such that each entry satisfies
\begin{align}\label{eq:2point_assume_main2}
 \E_{\omega}\, \sup_t |m_{jk}^{\omega}(t)|\leq C e^{-\mu K(|x_j-y_k|)},
\end{align}
for some constants $C,\mu>0$, and its matrix norm satisfies $\|M^{\omega}(t)\|\leq 1$ for any $\omega, t$.
Then we have the following determinant bound:
\begin{align}\label{eq:main2_1}
 \E_{\omega}\, \sup_t |\det M^{\omega}(t)| \leq C \left(\sum_{\substack{x\in X,\, y\in Y\\ |x-y|\geq D_m(X,Y)}} e^{-2\mu K(|x-y|)}\right)^{1/2}.
\end{align}
An alternate and stronger estimate is the following:
\begin{align}\label{eq:main2}
& \E_{\omega}\, \sup_t  |\det M^{\omega}(t)|\\
&\leq C \left(\sum_{\substack{(x_k, y_{\pi_0(k)})\notin \mathcal{C}}} e^{-2\mu K(|x_{j_0}-y_{\pi_0(k)}|)}+\sum_{\substack{(x_j, y_{\pi_0(j)})\in \mathcal{C}\\ (x_k, y_{\pi_0(k)})\notin \mathcal{C}}} e^{-2\mu K(|x_j-y_{\pi_0(k)}|)}\right)^{1/2}, \notag
\end{align}
where $\mathcal{C}$ is a cluster of points defined as in Definition \ref{def:cluster}, and $\pi_0$ is a minimal permutation as in Definition \ref{def:min_pi} and $j_0$ is defined in \eqref{def:j0}.
\end{theorem}

Our second main result concerns the $D_s$ distance.
\begin{theorem}\label{thm:determinant}
Let $X=\{x_1,...,x_n\}\subset \Z^d$ and $Y=\{y_1,...,y_n\}\subset \Z^d$. 
Let $M=(m_{jk})_{1\leq j,k\leq n}$ be a matrix with its each entry satisfying the following estimate
\begin{align}\label{eq:2point_pw}
|m_{jk}|\leq Ce^{-\mu |x_j-y_k|}.
\end{align}
Then we have 
\begin{align}\label{eq:est_in_thm_determinant}
\sum_{\pi \in S_n} \prod_{j=1}^n |m_{j\, \pi(j)}| \leq C_{\mu,d}^n\, e^{-\frac{\mu}{2\sqrt{d}} D_s(X,Y)}.
\end{align}
\end{theorem}
The estimate \eqref{eq:est_in_thm_determinant} of Theorem \ref{thm:determinant} clearly implies the bounds for both the determinant and the {\it permanent} of the matrix.
It would be interesting to see whether the result of Theorem \ref{thm:main2} also holds for permanents, where a Hadamard type inequality for permanent was proved \cite{CLL}.

Although Theorem \ref{thm:determinant} is of deterministic nature, it can be applied to disordered systems as well, see Theorems \ref{thm:ULE_loc} and \ref{thm:loc}.
This deterministic result can also be directly applied to show multi-point exponential decay of correlations for thermal states of single-particle systems, where the exponential decay of the two-point function \eqref{eq:2point_pw} was proved by Aizenman and Graf in \cite{AG}, see a detailed discussion of this application in Remark 2 after Theorem 1.1 in \cite{SW}.

\

For Gaussian fields, the multi-point correlation function can be written in terms of two-point functions through the (bononic) Wick rule.
Our third main result aims at providing an exponential decay bound for multi-point correlation function for systems satisfying the Wick rule, which in particular applies to the Ising model on $\Z^d$ with high temperatures, see Theorem \ref{thm:Ising_d>4}.

For a set $X=\{x_1,...,x_{2n}\}$ of $2n$ points configuration in $\Z^d$, $d\geq 1$, let $D_s(X)$ be defined as follows:
\begin{align}\label{def:DsX}
D_s(X)=\min_{\pi\in S_{2n}}\sum_{j=1}^n |x_{\pi(2j-1)}-x_{\pi(2j)}|.
\end{align}
One can also define similarly $D_m(X)$ with $\sum_{j=1}^n$ replaced with $\max_{j=1}^n$.

\begin{theorem}\label{thm:Pfaffian}
Let $X=\{x_1,...,x_{2n}\}\subset\Z^d$ be a set of $2n$ points configuration in $\Z^d$ with $d\geq 1$. 
Let $\{m_{jk}\}_{1\leq j,k\leq 2n}$ be such that
\begin{align}
|m_{jk}|\leq C e^{-\mu |x_j-x_k|},
\end{align}
for some constants $C,\mu>0$.
Then we have
\begin{align}
\frac{1}{n!}\sum_{\pi\in S_{2n}} \prod_{j=1}^n |m_{\pi(2j-1)\, \pi(2j)}|\leq C_{d,\mu}^n e^{-\frac{\mu}{2\sqrt{d}}D_s(X)}.
\end{align}
\end{theorem}
\begin{remark}
This theorem also applies to bound {\it Pfaffians} of matrices, see definition in \eqref{def:Pf}, hence also applies to systems that obey the fermionic Wick rule. Besides the applications to the Ising model discussed in Section 2, this theorem in particular formulates a higher dimensional version of Theorem 1.3 of \cite{SW} by Sims and Warzel, and moreover it replaces the $D_m$ distance in \cite{SW} with the stronger $D_s$ distance defined in \eqref{def:DsX}. 
A disordered version of this theorem will be addressed in a follow-up work \cite{HYPf}.
\end{remark}

Next we comment on the proofs of Theorems \ref{thm:main2},  \ref{thm:determinant} and \ref{thm:Pfaffian}.
In \cite{SW}, using crucially the ordering of points in $\Z^1$, the authors were able to rearrange the matrix $M$ into a bordered matrix form, and prove a technical core theorem that concerns the smallness of determinants (with respect to the $D_m$ distance) of such bordered matrices. 
In \cite{ABP}, for arbitrary dimension, the smallness of determinants (with respect to the Hausdorff distance $D_H$) is obtained from the smallness of an entire row or column, which is not true in general when replacing $D_H$ with the stronger $D_m$ or $D_s$.

The major difficulty in proving Theorem \ref{thm:main2} is the lack of effective ordering of points in higher dimensions.
We overcome this difficulty by developing a new approach in Lemma \ref{lem:cr1}, which proves the separation of a constructed cluster of $X$ points from the $Y$ points outside.
Another difficulty is to develop a different technical lemma concerning the smallness of determinant that fits our construction.
Our theorems \ref{thm:determinant} and \ref{thm:Pfaffian} concern the symmetrized distance $D_s$, and are built on a completely new combinatorics argument that essentially counts the number of paths that contribute distance between $2^k D_s$ and $2^{k+1}D_s$ for any $k\geq 0$.

The rest of the paper is organized as follows.
In Sections 2 and 3 we discuss the applications of our main results to the Ising model and the multi-point dynamical localization.
We then prove Theorem \ref{thm:main2} in Section 4, Theorem \ref{thm:determinant} in Section 5, and Theorem \ref{thm:Pfaffian} in Section 6. Sections 7 and 8 are devoted to the proofs of Theorem \ref{thm:Ising_d>4} and Theorems \ref{thm:ULE_loc}, \ref{thm:loc}.

In the following, a constant is allowed to change its value from line to line. 
We will only keep track of its dependence on the parameters.

\section{Application: multi-point correlation bound for the Ising model}\label{sec:apply_Ising}
One of the applications of our main theorems concerns the equilibrium states of the Ising model with ferromagnetic pair interactions.
The ferromagnetic Ising model on a graph $\Lambda$ is defined as
\begin{align}\label{def:Ising}
H=-\frac{1}{2}\sum_{\{x,y\}\subset \Lambda} J_{x,y} \sigma_x \sigma_y+h\sum_{x\in \Lambda}\sigma_x,
\end{align}
where $\{\sigma_x=\pm 1\}_{x\in\Lambda}$ are the spin variables, $h$ is the magnetic field and $J_{x,y}\geq 0$.

States of the system are given by probability measures on the space of configurations at the {\it inverse temperature} $\beta=1/T\geq 0$,
being described by the Gibbs states. 
For a finite graph $\Lambda$ the corresponding expectation value of functions of the spins are :
\begin{align*}
\langle f\rangle_{\Lambda,\beta}=\mathrm{tr}_{\Lambda} (f(\sigma) e^{-\beta H(\sigma)})/Z_{\Lambda}(\beta),
\end{align*}
where $\mathrm{tr}_{\Lambda}$ denotes $\prod_{x\in \Lambda} \frac{1}{2}\sum_{\sigma_x=\pm 1}$ and 
$Z_{\Lambda}(\beta)=\mathrm{tr}_{\Lambda} e^{-\beta H(\sigma)}$ is the partition function.
For infinite systems, the Gibbs states are defined by the limits of finite approximations.
For $\Lambda=\Z^d$, we shall write $\la f\ra_{\Z^d,\beta}$ as $\la f\ra_{\beta}$ for simplicity.

The graphs that we consider in this paper are $\Z^d$, $d\geq 2$, and their finite apporximants.
We focus on the nearest neighbor Ising model, meaning $J_{x,y}=\delta_1(|x-y|)$.
For $d\geq 2$, this model exhibits a line of first-order phase transitions for $h=0$ and $\beta\in (\beta_c,\infty)$, where $\beta_c>0$ is referred as the inverse of the critical temperature. Since the phase transition occurs at zero magnetic field, we restrict the discussion to $h=0$ and will omit h from the notation.

The correlation functions are defined through the following expectation values:
for $A\subset\Lambda$, let
\begin{align}\label{def:Ising_cf}
\langle \sigma_A\rangle_{\Lambda,\beta}=\langle \prod_{x\in A} \sigma_x\rangle_{\Lambda,\beta}.
\end{align}
It is known that the correlation functions $\la \sigma_A\ra_{\Lambda,\beta}$ are increasing as $\Lambda \uparrow \Z^d$.
The critical $\beta_c$ is such that 
\begin{align}\label{eq:betac}
\beta_c=\inf \{\beta: \la\sigma_x\ra_{\beta}>0\}.
\end{align}
The truncated two-point correlation function is defined by
\begin{align}\label{def:Ising_2p}
\la \sigma_x; \sigma_y\ra_{\Lambda,\beta}=\la \sigma_x\sigma_y\ra_{\Lambda,\beta}-\la \sigma_x\ra_{\Lambda,\beta} \la \sigma_y\ra_{\Lambda,\beta}.
\end{align}
Note that in the regime $0\leq \beta<\beta_c$, the truncated two-point function coincides with the (standard) two-point function defined in \eqref{def:Ising_cf}.

Exponential decay for the truncated two-point correlation function
\begin{align}\label{eq:exp_2p_Ising}
\la \sigma_x; \sigma_y\ra_{\beta}\leq e^{-\mu |x-y|}, \text{ for } x,y\in \Z^d
\end{align}
for the Ising model on $\Z^d$ has been proved by Aizenman-Barsky-Fernández in \cite{ABF} for $\beta<\beta_c$, and Duminil-Copin, Goswami and Raoufi in \cite{DCGR} for $\beta>\beta_c$.
Our goal is to use the exponential decay of the two-point functions as inputs to bound the multi-point functions, which is made possible due to the following remarkable results. 

In 1981, in his seminal work \cite{A81,A82} Aizenman proved that above the critical dimension $d=4$, the scaling limit of the multi-point correlation function follows the Wick rule.
The long standing problem about the critical dimension $d=4$ was recently tacked by Aizenman and Duminil-Copin in 2020 \cite{ADC}.
More precisely, the above mentioned results are:
\begin{theorem}\cite{A82, ADC}
For the nearest neightbor Ising model on $\Z^d$ with $d\geq 4$, if for some $\kappa(\delta)\to \infty$ the scaled correlation function converges (pointwise for $x_1,..,x_{2n}\in \R^d$), 
$$T_{2n}(x_1,...,x_{2n})=\lim_{\delta\to 0} \kappa(\delta)^{2n} \la \prod_{j=1}^{2n}\delta_{[x_j/\delta]}\ra_{\beta_c},$$
then the limiting function satisfies 
\begin{align}\label{eq:Gaussian}
T_{2n}(x_1,...,x_{2n})=\frac{1}{2^n n!}\sum_{\pi\in S_{2n}} \prod_{j=1}^n T_2(x_{\pi(2j-1)}, x_{\pi(2j)}).
\end{align}
\end{theorem}

Due to fact the the multi-point function can be written in terms of the two point function, we are able to use Theorem \ref{thm:Pfaffian} to obtain the following bound. Note that at $\beta=\beta_c$, two-point functions follows a power-law decay \cite{Simon}, hence Theorem \ref{thm:Pfaffian} does not apply.
\begin{theorem}\label{thm:Ising_d>4}
For the nearest neighbor ferromagnetic Ising model on $\Z^d$ with $d\geq 2$, for $0\leq \beta<\beta_c$, we have
\begin{align*}
\la \sigma_A\ra_{\beta}\leq C_{d,\mu}^n e^{-\frac{\mu}{2\sqrt{d}} D_s(A)},
\end{align*}
where $\mu$ is the constant in the two-point exponential decay \eqref{eq:exp_2p_Ising}, $C_{d,\mu}>0$ is a constant and $A\subset \Z^d$ is an arbitrary subset with $|A|=2n$.
\end{theorem} 
The proof, which relies heavily on the arguments in \cite{A82}, is presented in Section \ref{sec:Ising}.

\

Our second application is the multi-point correlation bound for points along a boundary for planar graphs.
In contrast to the Gaussian field structure of correlation function for points in the bulk of the lattice, the boundary correlation for 2D Ising model has a Pfaffian structure.
This was proved for graphs with a regular transfer matrix by Schultz-Mattis-Lieb \cite{SML}, and for any planar model by Groeneveld-Boel-Kasteleyn \cite{GBK}.
\begin{theorem}\cite{GBK}\label{thm:GBK}
Fix a planar graph $G$, arbitrary nearest-neighbor couplings $J$, and $\beta\geq 0$. Then, for any cyclically ordered 2n-tuple $(x_1,x_2,...,x_{2n})$ of sites located along the boundary of a fixed face of $G$, we have
\begin{align*}
\la \prod_{j=1}^{2n}\sigma_{x_j} \ra_{G,\beta}=\mathrm{Pf}([\la \sigma_{x_j} \sigma_{x_k}\ra_{G,\beta} ]_{1\leq j<k\leq 2n}).
\end{align*}
\end{theorem}
Above, a planar graph is a graph embedded in the plane $\R^2$ in such a way that its edges, depicted by bounded simple arcs, intersect only at their endpoints. The faces of the graph are the connected components of the plane minus the edges.
And the {\it Pfaffian} of a skew-symmetric matrix $M$ is defined by
\begin{align}\label{def:Pf}
\mathrm{Pf}(M):=\frac{1}{2^n n!} \sum_{\pi\in S_{2n}} \mathrm{sgn}(\pi)\, \prod_{j=1}^n m_{\pi(2j-1)\, \pi(2j)}.
\end{align}

As application of our Theorem \ref{thm:Pfaffian}, we obtain the following decay bound directly from the two-point decay \eqref{eq:exp_2p_Ising}.
\begin{theorem}\label{thm:Pfaffian_Z2}
Let $\Lambda\subset \Z^2$. Let $X=(x_1,...,x_{2n})$ be a cyclically ordered 2n-tuple located along the boundary of a fixed face of $\Lambda$. For $0\leq \beta<\beta_c$, let $\mu>0$ be the constant in the exponential decay of the two-point correlation function in \eqref{eq:exp_2p_Ising}, we have that for some constant $C_{d,\mu}>0$, the following holds
\begin{align*}
\la \prod_{j=1}^{2n} \sigma_{x_j} \ra_{\Lambda,\beta}\leq C_{d,\mu}^n e^{-\frac{\mu}{2\sqrt{d}}D_s(X)}.
\end{align*}
\end{theorem}

\section{Application: multi-point dynamical localization}
Another main motivation of studying the decay of multi-point correlation function is to study the multi-point dynamical localization that was introduced by Bravyi and König in \cite{BK}.
First, we introduce the notion of ergodic Schr\"odinger operators on $\Z^d$.
\begin{definition}\label{def:ergodic}
Let $(\Omega, \nu)$ be a probability measure space and $\E$ be its expectation. 
Let $T_j$, $j=1,...,d$, be a family of commuting ergodic maps on $\Omega$ with respect to $\nu$.
Let $\omega\to H_{\omega}$ be a strongly measurable map from $\Omega$ to self-adjoint operators on $\ell^2(\Z^d)$.
A family of self-adjoint operators $H_{\omega}$ on $\ell^2(\Z^d)$ is called ergodic if for each $x\in \Z^d$ there holds $H_{T^x \omega}=U_x H_{\omega}U_x^{-1}$ with $U_x\phi(y)=\phi(y-x)$, where $T^x:=T_1^{x_1}\cdots T_d^{x_d}$.
\end{definition}

Next, we introduce the definitions of two-point and multi-point dynamical localization.

\begin{definition}\label{def:TPDL}
A one-particle Hamiltonian $H_{\omega}$ on $\ell^2(\Z^d)$ is dubbed as {\it (two-point) exponential dynamical localized in expectation} (DLE) if there are constants $C,\mu>0$ such that for all $x,y\in \Z^d$, the following holds
\begin{align}\label{def:2p_DL}
\E_{\omega}\, \sup_t \left|\la \delta_{x}, e^{-itH_{\omega}} \delta_{y}\ra \right| \leq C e^{-\mu |x-y|}.
\end{align} 
\end{definition}

DLE has been proved for a large class of operators, including the Anderson model with i.i.d. random potentials with absolutly continuous distributions  in $\Z^1$ \cite{KS,DKS} and in $\Z^d$ for $d\geq 2$ \cite{A,AM,ASF}, the unitary Anderson model \cite{HJS}, the almost Mathieu operator on $\Z^1$ \cite{JKru}, random block operators \cite{ESS,CS}, an almost Mathieu type operator on $\Z^d$ \cite{GYZ}, and random word models \cite{Nishant}.
If one replaces the exponential decay on the right-hand side of \eqref{def:2p_DL} with a power-law decay $|x-y|^{-\mu}$ with $\mu>0$, then the (power-law) DLE was proved for a larger class of operators including one dimensional Anderson model with singular supports \cite{CKM,DS} and a one-dimensional long range almost Mathieu type operator \cite{BJ}.

If one takes 
$$m_{jk}^{\omega}(t)=\la \delta_{x_j}, e^{-itH_{\omega}}\delta_{y_k}\ra,$$
then clearly DLE in \eqref{def:2p_DL} verifies exactly the assumption \eqref{eq:2point_assume_main2} in Theorem \ref{thm:main2}.
Thus, as our first application, Theorem \ref{thm:main2} yields a multi-point dynamical localization type result in the following sense for all the (two-point) dynamical localized systems, with $D_s$ distance replaced with the $D_m$ distance.

\begin{definition}\label{def:MPDL}
A one-particle Hamiltonian $H_{\omega}$ on $\ell^2(\{1,...,N\})$ is dubbed as {\it multi-point dynamical localized (MPDL) in expectation}, see \cite{BK}, if there are constants $C,\mu>0$ such that for all $n\leq N$ sufficiently large, the following holds
\begin{align}\label{def:mp_DL}
\sup_{t\in \R} \E_{\omega}\, \left|\det(\la \delta_{x_j}, e^{-itH_{\omega}} \delta_{y_k}\ra)_{1\leq j,k\leq n}\right| \leq C^n e^{-\mu N},
\end{align} 
for all configurations $X,Y$ with $D_s(X,Y)\geq N/8$.
\end{definition}
In \cite{BK}, MPDL in expectation (MPDLE) was {\it used as an assumption} to study the growth of storage time of quantum memory.
It was conjectured \cite{BK} that MPDLE holds in the regime of strong disorder, but there has been no rigorous proof for any model so far, see also \cite{CS}.
It was also proposed as an open question in \cite{SW} that whether MPDLE always follows from DLE.

Another achievement of this paper is to prove the existence of MPDLE models as an application of our Theorem \ref{thm:determinant}.  
We prove MPDLE (actually a stronger uniform localization result) for systems with uniformly localized eigenfunctions (see definition below) in Theorem \ref{thm:ULE_loc}.
Hence we provide {\it the first examples} of such feature, and in particular verifies the assumptions made in \cite{BK} for these models. 

Regarding the question whether MPDLE always follows from DLE: due to the nature of the sum distance involved in $D_s$ (rather than the maximal distance in $D_m$), MPDLE seems to require the control of {\it the expectation of each term} in the determinant expansion, which will unlikely follow from the general DLE.
But we are still able to prove an almost sure version of MPDL in Theorem \ref{thm:loc}, rather than MPDLE, for general DLE systems.

At last, we comment that, assuming {\it the expectation of each term} in the determinant expansion is controlled as in \eqref{eq:stronger_assume} below, which is {\it a stronger assumption} than DLE,
\begin{align}\label{eq:stronger_assume}
\E \left[\prod_{j=1}^n \left|\la \delta_{x_j}, e^{-itH_{\omega}} \delta_{y_{\pi(j)}}\ra\right|\right]\leq C^n e^{-\mu\sum_{j=1}^n |x_j-y_{\pi(j)}|},
\end{align}
MPDLE was proved (conditionally) in the $d=1$ case in \cite{BK}, see Lemma 3 therein. We can strengthen this {\it conditional result} to arbitrary dimension (the proof is the same as that for Theorem \ref{thm:determinant}, so we omit it in this paper).
It is an interesting question if \eqref{eq:stronger_assume} is true in general.

Next, we explain the applications of Theorem \ref{thm:determinant} mentioned above in details.
We introduce the definition of operators with uniformly localized eigenfunctions.
\begin{definition}\label{def:ULE}
We say a self-adjoint operator $H$ on $\ell^2(\Z^d)$ has uniformly localized eigenfunctions (ULE) if $H$ has a complete set $\{\phi_k\}_{k=1}^{\infty}$ of orthonormal eigenfunctions, and there exists $C, \mu>0$ and $m_k\in \Z^d$ such that 
\begin{align*}
|\phi_k(m)|\leq C e^{-\mu |m-m_k|}, 
\end{align*}
holds for any $k\in \N$ and $m\in \Z^d$.
\end{definition}
For ULE systems, we have
\begin{theorem}\label{thm:ULE_loc}
If $H_{\omega}$ is a family of ergodic operators with ULE for a positive measure set of $\omega$, 
Namely, for positive measure set of $\omega$, $H_{\omega}$ has a complete set $\{\phi_k^{\omega}\}_{k=1}^{\infty}$ of orthonormal eigenfunctions, and there exist $C_{\omega}, \mu_{\omega}>0$ and $m_k^{\omega}\in \Z^d$ such that
\begin{align*}
|\phi_k^{\omega}(m)|\leq C_{\omega} e^{-\mu_{\omega} |m-m_k^{\omega}|}.
\end{align*}
Then multi-point dynamical localization holds uniformly for $H_{\omega}$ for any $\omega\in \mathrm{supp}(\nu)$. Namely, for some constants $C,\mu>0$,
\begin{align*}
\sup_{t\in \R} \left|\det (\la \delta_{x_j}, e^{-itH_{\omega}} \delta_{y_k}\ra)_{1\leq j,k\leq n}\right|\leq C^n e^{-\mu D_s(X,Y)},
\end{align*}
holds for any $\omega\in \mathrm{supp}(\nu)$ and arbitrary $n$-point configurations $X,Y\subset \Z^d$.
\end{theorem}
\begin{remark}
Although it was pointed out in \cite{J,DJLS} that ULE does not hold for the Anderson model or the almost Mathieu operator, as it leads to a violation of generic continuous spectrum. 
ULE has been proved for some limit-periodic operators \cite{DG,DG2} and some quasi-periodic operators \cite{GFP,BLS,JKmon,Kmon,Chu}.
\end{remark}

For general dynamical localized disordered systems, we have:
\begin{theorem}\label{thm:loc}
Let $H_{\omega}$ be a family of ergodic operators. 
Let $I\subset \C$ and consider a family of operators $\rho(s,t,H_{\omega})$ on $\ell^2(\Z^d)$ for $s,t\in I$ that exhibit localization in the sense that for some constants $C,\mu\in (0,\infty)$ and for all $x,y\in \Z^d$, we have
\begin{align}\label{eq:2point_E}
\E_{\omega}\, \left[ \sup_{s,t\in I} |\la \delta_x, \rho(s,t, H_{\omega})\delta_y\ra | \right]\leq Ce^{-\mu |x-y|}.
\end{align}
Then for a.e. $\omega$, and any $n\in \N$, and any pair of configurations $X=\{x_1,...,x_n\}\subset \Z^d$, $Y=\{y_1,...,y_n\}\subset \Z^{d}$,
we have that
\begin{align}\label{eq:loc_as}
\sup_{s,t\in I^n} |\det(\la \delta_{x_j}, \rho(s_j, t_k, H_{\omega}) \delta_{y_k}\ra)| \leq C_{\omega,d,\mu}^n \left(\prod_{j=1}^n (1+|x_j|)^{d+1}\right)\, e^{-\frac{\mu}{4\sqrt{d}} D_s(X,Y)},
\end{align}
where without loss of generality we have assumed $\prod_{j=1}^n (1+|x_j|)\leq \prod_{j=1}^n (1+|y_j|)$.
\end{theorem}
The factor involving $|x_j|$'s emerges when we convert the averaged assumption \eqref{eq:2point_E} into a deterministic assumption of the form \eqref{eq:2point_pw}.
Assuming $|x_j|\leq N$ for all $j$, we have
$$\sup_{s,t\in I^n} |\det(\la \delta_{x_j}, \rho(s_j, t_k, H_{\omega}) \delta_{y_k}\ra)| \lesssim_{\omega,d,\mu} N^{dn}e^{-\frac{\mu}{4} D_s(X,Y)},$$
which is effective when $D_s\gtrsim_{\omega,d,\mu} n\log N$.

When $\rho(s,t,H_{\omega})=e^{-itH_{\omega}}$, \eqref{eq:2point_E} is the same as DLE.

\section{Proof of Theorem \ref{thm:main2}}
Throughout this section we shall denote $D_m(X,Y)$ by $D$ for simplicity.
The idea to overcome the lack of effective ordering in higher dimensions is to prove that there is a cluster of points that is gapped from the points outside.
We prove this by utilizing a notation of minimal permutation (see Definition \ref{def:min_pi} below) and a novel construction that violates such minimality if the cluster is not gapped from its complement.

For a permutation $\pi$, let $D_{\pi}:=\max_{j=1}^n |x_j-y_{\pi(j)}|$.
For a given minimal $\pi$ in the sense that $D_{\pi}=D$, let $\mathcal{N}(\pi)$ be the number of $j$'s such that $|x_j-y_{\pi(j)}|=D$. Namely $\mathcal{N}(\pi)$ counts the number of times when the maximal value $D$ is attained in a permutation $\pi$.

\begin{definition}\label{def:min_pi}
We say $\pi_0$ is a minimal permutation if $D_{\pi_0}=D_m(X,Y)$ and $\pi_0$ also minimizes $\mathcal{N}(\pi)$. 
Namely, for any arbitrary $\pi$ such that $D_{\pi}=D_m(X,Y)$, we always have $\mathcal{N}(\pi_0)\leq \mathcal{N}(\pi)$.
\end{definition}
For such a $\pi_0$, let $j_0\in \{1,...,N\}$ be such that 
\begin{align}\label{def:j0}
|x_{j_0}-y_{\pi_0(j_0)}|=D.
\end{align} 
Note that there may be multiple choices for $\pi_0$ (and $j_0$), in which case we simply choose one of them.

We define a cluster $\mathcal{C}$ of points as follows:
\begin{definition}\label{def:cluster}
We say $(x_z,y_{\pi_0(z)})\in \mathcal{C}$ if $|x_{j_0}-y_{\pi_0(z)}|<D$ or there exists a chain of distinct points $\{(x_{a_k},y_{\pi_0(a_k)})\}_{k=1}^m$ (with $m\geq 1$) such that $j_0, z\notin \{a_1,...,a_k\}$ and $|x_{j_0}-y_{\pi_0(a_1)}|<D$, $|x_{a_m}-y_{\pi_0(z)}|<D$ and $|x_{a_k}-y_{\pi(a_{k+1})}|<D$ for any $1\leq k\leq m-1$.
\end{definition}

We have the following lemmas.
\begin{lemma}\label{lem:cr1}
We have $(x_{j_0}, y_{\pi_0(j_0)})\notin \mathcal{C}$. Also for any $(x_w, y_{\pi_0(w)})\notin \mathcal{C}$ 
and any $(x_z, y_{\pi_0(z)})\in \mathcal{C}$, we have
\begin{align}\label{eq:xL-yR}
|x_z-y_{\pi_0(w)}|\geq D, \text{ and } |x_{j_0}-y_{\pi_0(w)}|\geq D.
\end{align}
\end{lemma}
\begin{proof}
The claim \eqref{eq:xL-yR} is a direct consequence of the definition of $\mathcal{C}$. We now prove that $(x_{j_0}, y_{\pi_0(j_0)})\notin \mathcal{C}$.
Suppose otherwise, by the definition of $\mathcal{C}$, there exists a chain of points $\{(x_{a_k},y_{\pi_0(a_k)})\}_{k=1}^m$ such that:
\begin{align}\label{eq:chain1}
\begin{cases}
|x_{j_0}-y_{\pi_0(a_1)}|<D\\
|x_{a_k}-y_{\pi_0(a_{k+1})}|<D, \text{ for } 1\leq k\leq m-1\\
|x_{a_m}-y_{\pi_0(j_0)}|<D
\end{cases}
\end{align}
In this case, we can define a permutation $\pi_1$ such that
\begin{align*}
\begin{cases}
\pi_1(j_0)=\pi_0(a_1)\\
\pi_1(a_k)=\pi_0(a_{k+1})\, \text{ for } 1\leq k\leq m-1\\
\pi_1(a_m)=\pi_0(j_0)\\
\pi_1(r)=\pi_0(r), \text{ for } r\notin \{j_0,a_1,...,a_m\}=:U.
\end{cases}
\end{align*}
Clearly by \eqref{eq:chain1}, we have
\begin{align*}
D_{\pi_1}=\max_{r=1}^n |x_r-y_{\pi_1(r)}|
=&\max(\max_{r\in U} |x_r-y_{\pi_1(r)}|, \max_{r\in U^c} |x_r-y_{\pi_1(r)}|)\leq D
\end{align*}
and when $D_{\pi_1}=D$, we must have $\mathcal{N}_{\pi_1}<\mathcal{N}_{\pi_0}$. This violates the minimality of $\pi_0$, hence leads to a contradiction. This proves the claimed result.
\end{proof}

We also need the following, which plays the same role as the technical core Theorem 3.1 in \cite{SW}.
\begin{lemma}\label{lem:la}
Let $M\in \C^{p\times p}$ be a block matrix, satisfying $\|M\|\leq 1$, and is of the following form:
\begin{align}\label{eq:blockM}
M=\left(\begin{matrix}
A\ & B\\
C\ & D
\end{matrix}\right)
\end{align}
where $A\in \C^{\ell\times m}$ with $m<\ell$. Then
\begin{align*}
|\det M|\leq \|B\|.
\end{align*}
\end{lemma}
\begin{proof}
Since $m<\ell$, there exists a normalized vector $\psi_1\in \C^{\ell}$ such that
\begin{align}\label{eq:psi1}
\psi_1^T A=0.
\end{align}
Picking another $m-1$ normalized vectors $\{\psi_2,...,\psi_{\ell}\}\subset \C^{\ell}$ such that $\{\psi_1,...,\psi_{\ell}\}$ forms an orthonormal set. 
Let $U=(\psi_1, \psi_2,..., \psi_{\ell})\in \C^{\ell\times \ell}$, then $U$ is a unitary matrix.

Let $R$ be a block matrix defined as
\begin{align*}
R=\left(\begin{matrix}
U^T\ \  & 0\\
0\ \  & I_{(p-\ell)\times (p-\ell)}
\end{matrix}\right).
\end{align*}
We have
\begin{align*}
\tilde{M}:=RM=
\left(\begin{matrix}
U^T A\ &U^T B\\
C & D
\end{matrix}\right),
\end{align*}
and
\begin{align*}
\|\tilde{M}\|\leq \|M\|\cdot \|R\|=\|M\|, \text{ and } |\det \tilde{M}|=|\det M|.
\end{align*}
By Hadamard's inequality, we have
\begin{align}\label{eq:ha1}
|\det \tilde{M}|\leq \prod_{j=1}^p \|v_j\|\leq \|v_1\|,
\end{align}
where $\{v_1,...,v_p\}$ are the rows of $\tilde{M}$ and we used that $\|v_j\|\leq \|\tilde{M}\|\leq 1$ for $2\leq j\leq p$.
Note that the first row of $U^T A=\psi_1^T A=0$ by \eqref{eq:psi1}. Hence
\begin{align}\label{eq:ha2}
\|v_1\|=\|\psi_1^T B\|\leq \|B\|.
\end{align}
Combining \eqref{eq:ha1} with \eqref{eq:ha2}, we have proved Lemma \ref{lem:la}.
\end{proof}

Now we are ready to combine Lemmas \ref{lem:cr1} and \ref{lem:la} to complete the proof of Theorem \ref{thm:main2}. 
Let 
\begin{align*}
\begin{cases}
\mathcal{R}_A:=\{(j, \pi_0(k)): (x_k, y_{\pi_0(k)})\in \mathcal{C} \text{ and either } j=j_0 \text{ or } (x_j, y_{\pi_0(j)})\in \mathcal{C}\}\\
\mathcal{R}_B:=\{(j, \pi_0(k)): (x_k, y_{\pi_0(k)})\notin \mathcal{C} \text{ and either } j=j_0 \text{ or } (x_j, y_{\pi_0(j)})\in \mathcal{C}\}\\
\mathcal{R}_C:=\{(j, \pi_0(k)): j\neq j_0 \text{ and } (x_j,y_{\pi_0(j)})\notin \mathcal{C}, \text{ and } (x_k, y_{\pi_0(k)})\in \mathcal{C}\}\\
\mathcal{R}_D:=\{(j, \pi_0(k)): j\neq j_0 \text{ and } (x_j,y_{\pi_0(j)})\notin \mathcal{C}, \text{ and } (x_k, y_{\pi_0(k)})\notin \mathcal{C}\}
\end{cases}
\end{align*}
By switching rows and columns we arrange the elements of $M^{\omega}(t)=(m_{j\, \pi_0(k)}^{\omega}(t))$ into the form in \eqref{eq:blockM}, where
\begin{align*}
m_{j\, \pi_0(k)}^{\omega}(t)\in E\, \text{ iff } (j, \pi_0(k))\in \mathcal{R}_{E}, \text{ for } E=A,B,C,D.
\end{align*}
Clearly $A\in \C^{(m+1)\times m}$, where $m=|\mathcal{C}|$. 
Hence by Lemma \ref{lem:la}, we have
\begin{align*}
|\det M^{\omega}(t)|\leq \|B\|.
\end{align*}

For a matrix $B$, let $\|B\|_F$ be the Frobenius norm.
Using the two-point assumption \eqref{eq:2point_assume_main2} and convexity, we have
\begin{align}\label{eq:main2_2}
\E_{\omega}\, \sup_t \|B\|\leq \E_{\omega}\, \sup_t \|B\|_{F}
\leq &C \left(\sum_{(j, \pi_0(k))\in \mathcal{R}_B} e^{-2\mu K(|x_j-y_{\pi_0(k)}|)}\right)^{1/2}.
\end{align}
This proves \eqref{eq:main2} by revoking the definition of $\mathcal{R}_B$.

Noting that for each $|x_j-y_{\pi_0(k)}|$ on the right-hand side of \eqref{eq:main2_2}, we have by Lemma \ref{lem:cr1} that
$$|x_j-y_{\pi_0(k)}|\geq D.$$ 
We can bound \eqref{eq:main2_2} to prove \eqref{eq:main2_1}.
\qed

\section{Proof of Theorem \ref{thm:determinant}}
It suffices to prove the following lemma.
\begin{lemma}\label{lem:determinant}
For some constant $C_{d,\mu}>0$, there holds
\begin{align*}
\sum_{\pi\in S_n} e^{-\mu \sum_{j=1}^n |x_j-y_{\pi(j)}|}\leq C_{d,\mu}^n e^{-\frac{\mu}{2\sqrt{d}}D_s(X,Y)}.
\end{align*}
\end{lemma}

Throughout this section, we denote $D_s(X,Y)$ by $D$ for simplicity.
In this section, for arbitrary two points $x,y\in \Z^d$, we denote $|x-y|_{\infty}:=\max_{j=1}^d |x_j-y_j|$.
We introduce $\tilde{D}_s(X,Y)$ similar to $D_s(X,Y)$ but with the Euclidean norms replaced with the sup norms:
\begin{align*}
\tilde{D}_s(X,Y):=\min_{\pi\in S_n} \sum_{j=1}^n |x_j-y_{\pi(j)}|_{\infty}.
\end{align*}
Throughout this section we shall denote $\tilde{D}_s(X,Y)$ by $\tilde{D}$ for simplicity.
Clearly, we have
\begin{align}\label{eq:D_tildeD}
\tilde{D}\leq D\leq \sqrt{d}\, \tilde{D}.
\end{align}
The reason we introduce $\tilde{D}$ is that it only takes integer value for $X,Y\subset \Z^d$, which plays a crucial role in our counting argument in Lemma \ref{lem:counting}.

Let $B_{d,\mu}\geq 1$ be the smallest number such that
\begin{align}\label{eq:D>1}
(3x)^{d}\leq e^{\frac{\mu}{2}x}, \text{ for all } x\geq B_{d,\mu}.
\end{align}

Let $\pi_0$ be a permutation such that
\begin{align*}
\sum_{j=1}^n |x_j-y_{\pi_0(j)}|_{\infty}=\min_{\pi}\sum_{j=1}^n |x_j-y_{\pi(j)}|_{\infty}=\tilde{D}.
\end{align*}
In case of there are multiple $\pi_0$ that attains the minimum value $\tilde{D}$, we will simply choose one of them.
The following counting lemma is the key to the proof of the determinant bound.
\begin{lemma}\label{lem:counting}
For any integer $\ell \geq \tilde{D}$, let $\mathcal{M}_{\ell}$ be 
\begin{align*}
\mathcal{M}_{\ell}:=\{\pi\in S_n: \sum_{j=1}^n |x_j-y_{\pi(j)}|_{\infty}=\ell\}.
\end{align*}
Then we have the following estimate on the cardinality of $M_{\ell}$.
\begin{align*}
|\mathcal{M}_{\ell}|\leq C_d^n \left(\frac{2\ell+n}{n}\right)^{dn}.
\end{align*}
\end{lemma}
\begin{proof}
For $\ell \in \N$,
let 
\begin{align}\label{def:Sl}
\mathcal{S}_{\ell}:=\{r=(r_1,...,r_n)\in \N^n: \sum_{j=1}^n r_j=\ell.\}
\end{align}
We have, by Stirling's formula
\begin{align}\label{eq:Sl_1}
|\mathcal{S}_{\ell}|={\ell+n-1 \choose n-1}\leq \frac{C}{\sqrt{n}} e^n \left(\frac{\ell+n}{n}\right)^n.
\end{align}
For each $r\in \mathcal{S}_{\ell}$, let 
\begin{align}\label{def:Tr}
\mathcal{T}_r:=\{\pi: |x_j-y_{\pi(j)}|_{\infty}=r_j, \text{ for each } j\in \{1,...,n\}\}.
\end{align}
Note that there exists at most $2d\cdot (2r_j+1)^{d-1}$ many integer point $z$'s in $\Z^d$ such that 
\begin{align*}
|x_j-z|_{\infty}=r_j.
\end{align*}
Hence there are at most $2d\cdot (2r_j+1)^{d-1}$ many possible values for $\pi(j)$ such that $$|x_j-y_{\pi(j)}|_{\infty}=r_j.$$
Thus for each $r\in \mathcal{S}_{\ell}$, we bound the geometric mean of $2r_j+1$ by the arithmetic mean and obtain that
\begin{align}\label{eq:Tr1}
|\mathcal{T}_r|\leq (2d)^n \prod_{j=1}^n (2r_j+1)^{d-1}\leq (2d)^n \left(\frac{2\sum_{j=1}^n r_j+n}{n}\right)^{(d-1)n}=(2d)^n \left(\frac{2\ell+n}{n}\right)^{(d-1)n}.
\end{align}
Combining \eqref{eq:Sl_1} with \eqref{eq:Tr1}, we have
\begin{align*}
|\mathcal{M}_{\ell}|
=|\mathcal{S}_{\ell}|\cdot |\mathcal{T}_r|
\leq &(2 d e)^n \left(\frac{2\ell+n}{n}\right)^{dn}.
\end{align*}
This proves the lemma.
\end{proof}
Now we are in the position to prove Theorem \ref{thm:determinant}.
We have by Lemma \ref{lem:counting},
\begin{align}\label{eq:Ml}
\sum_{\pi\in S_n} \prod_{j=1}^n |m_{j\, \pi(j)}|
\leq &C^n \sum_{\pi\in S_n} e^{-\mu \sum_{j=1}^n |x_j-y_{\pi(j)}|}\\
\leq &C^n \sum_{\pi\in S_n} e^{-\mu \sum_{j=1}^n |x_j-y_{\pi(j)}|_{\infty}}\\
\leq &C^n \sum_{\ell=\tilde{D}}^{\infty} \sum_{\pi\in \mathcal{M}_{\ell}} e^{-\mu\sum_{j=1}^n |x_j-y_{\pi(j)}|_{\infty}} \notag\\
\leq &C_d^n \sum_{\ell=\tilde{D}}^{\infty} \left(\frac{2\ell+n}{n}\right)^{dn} e^{-\mu \ell}.
\end{align}
Let $B_{d,\mu}$ be defined as in \eqref{eq:D>1}.
Next we distinguish two cases:

{\it Case 1.} If $\tilde{D}<B_{d,\mu}n$.
We split the sum in \eqref{eq:Ml} into
\begin{align}\label{eq:MI3}
\sum_1:=\sum_{\ell=\tilde{D}}^{B_{d,\mu} n} C_d^n \left(\frac{2\ell+n}{n}\right)^{dn} e^{-\mu \ell},
\end{align}
and 
\begin{align}\label{eq:MI4}
\sum_2:=\sum_{\ell>B_{d,\mu} n} C_d^n \left(\frac{2\ell+n}{n}\right)^{dn} e^{-\mu \ell},
\end{align}
To estimate $\sum_1$, we bound $\ell$ by $B_{d,\mu}n$, which yields
\begin{align}\label{eq:MI6}
\sum_1\leq \sum_{\ell=D}^{B_{d,\mu}n}C_d^n (2B_{d,\mu}+1)^{dn} e^{-\mu \ell}\leq C_{d,\mu}^n e^{-\mu \tilde{D}}.
\end{align}

Next, we estimate $\sum_2$. Since $\ell\geq B_{d,\mu}n\geq n$, we have by \eqref{eq:D>1},
\begin{align}\label{eq:MI2}
\left(\frac{2\ell+n}{n}\right)^{dn}\leq \left(\frac{3\ell}{n}\right)^{dn}\leq e^{\frac{\mu}{2}\ell}.
\end{align}
Hence 
\begin{align}\label{eq:MI5}
\sum_2\leq \sum_{\ell>B_{d,\mu}n} C_d^n e^{-\frac{\mu}{2}\ell}\leq C_d^n e^{-\frac{\mu}{2}\tilde{D}}.
\end{align}

{\it Case 2.} If $\tilde{D}\geq B_{d,\mu}n$. Since $\ell\geq \tilde{D}\geq B_{d,\mu}n\geq n$, we can argue as in \eqref{eq:MI2} to bound the sum in \eqref{eq:Ml} as follows
\begin{align}\label{eq:MI8}
\sum_{\pi\in S_n} \prod_{j=1}^n  |m_{j\, \pi(j)}| \leq \sum_{\ell\geq \tilde{D}} C_d^n e^{-\frac{\mu}{2}\ell}\leq C_d^n e^{-\frac{\mu}{2}\tilde{D}}.
\end{align}

Combining \eqref{eq:Ml} with \eqref{eq:MI6}, \eqref{eq:MI5}, \eqref{eq:MI8} and \eqref{eq:D_tildeD}, we have proved Theorem \ref{thm:determinant}. \qed

\section{Proof of Theorem \ref{thm:Pfaffian}}
Throughout this section, we denote $\{1,2,...,2n\}=:[2n]$.
It is easy to see that $D_s(X)$ has the following alternate representation in terms of the $D_s$ distance that we introduced earlier:
\begin{align}\label{eq:DsX=DsXY}
D_s(X)=\min_{\substack{Y\subset X\\ |Y|=n}} D_s(X\setminus Y, Y).
\end{align}

Let 
\begin{align}\label{def:A2n}
\mathcal{A}_{2n}:=\{\pi\in S_{2n}: \pi(1)<\pi(3)<\cdots <\pi(2n-1)\}.
\end{align}
This set can be decomposed as a disjoint union as follows:
\begin{align}
\mathcal{A}_{2n}=\sqcup_{B=\{b_1<b_2<...<b_n\}\subset [2n]} \mathcal{A}_B,
\end{align}
where 
\begin{align}
\mathcal{A}_B=\{\pi\in S_{2n}: \pi(2j-1)=b_j\, \text{ for } 1\leq j\leq n\}.
\end{align}
Let $Y_B\subset X$ be such that
\begin{align}
Y_B:=\{x_{b_1}<x_{b_2}<...<x_{b_n}\}.
\end{align}
Next we estimate
\begin{align}\label{eq:Pf1}
\frac{1}{n!}\sum_{\pi\in S_{2n}} \prod_{j=1}^n |m_{\pi(2j-1)\, \pi(2j)}|
=&\sum_{\pi \in \mathcal{A}_{2n}} \prod_{j=1}^n |m_{\pi(2j-1)\, \pi(2j)}|\\
\leq &C^n \sum_{\pi\in \mathcal{A}_{2n}} e^{-\mu \sum_{j=1}^{n} |x_{\pi(2j-1)}-x_{\pi(2j)}|}\\
=&C^n\sum_{\substack{B\subset [2n]\\ |B|=n}} \sum_{\pi\in \mathcal{A}_B} e^{-\mu \sum_{j=1}^{n} |x_{b_j}-x_{\pi(2j)}|}.
\end{align}
By Lemma \ref{lem:determinant} and \eqref{eq:DsX=DsXY}, we have
\begin{align}\label{eq:Pf2}
\sum_{\pi\in \mathcal{A}_B} e^{-\mu \sum_{j=1}^{n} |x_{b_j}-x_{\pi(2j)}|}\leq C_{d,\mu}^n e^{-\frac{\mu}{2\sqrt{d}} D_s(X\setminus Y_B, Y_B)}
\leq C_{d,\mu}^n e^{-\frac{\mu}{2\sqrt{d}} D_s(X)}.
\end{align}
Hence combining \eqref{eq:Pf1} with \eqref{eq:Pf2}, and using the Stirling's formula, we have
\begin{align}
\frac{1}{n!}\sum_{\pi\in S_{2n}} \prod_{j=1}^n |m_{\pi(2j-1)\, \pi(2j)}|
\leq &C_{d,\mu}^n \cdot |\{B: B\subset [2n], |B|=n\}| \cdot e^{-\frac{\mu}{2\sqrt{d}} D_s(X)}\\
\leq &C_{d,\mu}^n e^{-\frac{\mu}{2\sqrt{d}} D_s(X)},
\end{align}
which proves Theorem \ref{thm:Pfaffian}. \qed

\section{Proof of Theorem \ref{thm:Ising_d>4}: multi-point correlation bound for the Ising model}\label{sec:Ising}
For simplicity, in this section we shall write $\la \cdot \ra_{\beta}$ as $\la \cdot \ra$.
The truncated (four-point) correlation function $U_4$ is the following:
\begin{align*}
U_4(x_1,x_2,x_3,x_4)=\la \sigma_{x_1}\sigma_{x_2}\sigma_{x_3}\sigma_{x_4}\ra-
(\la \sigma_{x_1}\sigma_{x_2}\ra \la\sigma_{x_3}\sigma_{x_4}\ra+\la\sigma_{x_1}\sigma_{x_3}\ra \la\sigma_{x_2}\sigma_{x_4}\ra+\la \sigma_{x_1}\sigma_{x_4}\ra \la \sigma_{x_2}\sigma_{x_3}\ra)
\end{align*}
Proposition 5.3 in \cite{A82} provides a tree diagram bound for $U_4$ as follows.
\begin{proposition}\cite{A82}
In a ferromagnetic system, for any four points, we have
\begin{align*}
|U_4(x_1,x_2,x_3,x_4)|\leq 2\sum_{y\in \Z^d} \la \sigma_{x_1}\sigma_y\ra \la \sigma_{x_2}\sigma_y\ra \la \sigma_{x_3}\sigma_y\ra \la \sigma_{x_4}\sigma_y\ra.
\end{align*}
\end{proposition}
As an easy corollary we have
\begin{align}\label{eq:U4}
|U_4(x_1,x_2,x_3,x_4)|\leq C_d e^{-\frac{\mu}{2}(|x_1-x_2|+|x_3-x_4|)}.
\end{align}
For any $2n$ points, let 
\begin{align*}
G_{2n}(x_1,...,x_{2n}):=\frac{1}{2^n} \sum_{\pi\in \mathcal{A}_{2n}}\prod_{j=1}^n \la \sigma_{x_{\pi(2j-1)}}\sigma_{x_{\pi(2j)}}\ra,
\end{align*}
where $\mathcal{A}_{2n}$ is defined in \eqref{def:A2n}.
Proposition 12.1 in \cite{A82} provides an error estimate \footnote{\cite{A82} provides both upper and lower bounds, but we only need the upper bound here.} as follows.
\begin{proposition}\cite{A82}
In a ferromagnetic Ising system, 
\begin{align*}
&|\la \sigma_{x_1}\cdots \sigma_{x_{2n}}\ra-G_{2n}(x_1,...,x_{2n})|\\
\leq &\sum_{1\leq j_1<j_2<j_3<j_4\leq 2n} |U_4(x_{j_1},x_{j_2},x_{j_3},x_{j_4})|
G_{2n-4}(...,\check{x}_{j_1},...,\check{x}_{j_2},...,\check{x}_{j_3},...,\check{x}_{j_4},...),
\end{align*}
where $\check{}$ indicates an omitted site.
\end{proposition}

By our Theorem \ref{thm:Pfaffian}, with $m_{jk}:=\la \sigma_{x_j} \sigma_{x_k} \ra$, we have
\begin{align}\label{eq:G2n}
G_{2n}(x_1,...,x_{2n})\leq C_{d,\mu}^n e^{-\frac{\mu}{2\sqrt{d}}D_s(X)},
\end{align}
and
\begin{align}\label{eq:G2n-4}
G_{2n-4}(...,\check{x}_{j_1},...,\check{x}_{j_2},...,\check{x}_{j_3},...,\check{x}_{j_4},...)\leq C_{d,\mu}^n e^{-\frac{\mu}{2\sqrt{d}}D_s(X\setminus \{x_{j_1},x_{j_2},x_{j_3},x_{j_4}\})}.
\end{align}
Combining \eqref{eq:G2n-4} with \eqref{eq:U4} we have
\begin{align}\label{eq:G2n-4_U4}
&|U_4(x_{j_1},x_{j_2},x_{j_3},x_{j_4})| G_{2n-4}(...,\check{x}_{j_1},...,\check{x}_{j_2},...,\check{x}_{j_3},...,\check{x}_{j_4},...) \notag\\
\leq &C_{d,\mu}^n e^{-\frac{\mu}{2}(|x_{j_1}-x_{j_2}|+|x_{j_3}-x_{j_4}|)}e^{-\frac{\mu}{2\sqrt{d}} D_s(X\setminus \{x_{j_1},x_{j_2},x_{j_3},x_{j_4}\})} \notag \\
\leq &C_{d,\mu}^n e^{-\frac{\mu}{2\sqrt{d}} D_s(X)}.
\end{align}
Theorem \ref{thm:Ising_d>4} then follows from combining \eqref{eq:G2n} with \eqref{eq:G2n-4_U4}.  \qed

\section{Proof of Multi-point dynamical localization}

\subsection{Proof of Theorem \ref{thm:ULE_loc}}
The following notion of homegeneous ULE was introduced in \cite{uniform}.
\begin{definition}\label{def:HULE}
$H_{\omega}$ has homogeneous ULE in a set $S$ means that $H_{\omega}$ has ULE for any $\omega\in S$ and
\begin{align*}
|\phi_n^{\omega}(m)|\leq C e^{-\mu |m-m_n^{\omega}|},
\end{align*}
for some constants $C,\mu>0$ independent of $\omega$.
\end{definition}
The following theorem was proved in \cite{uniform}:
\begin{theorem}\label{thm:uniform}
If $H_{\omega}$ has ULE for $\omega$ in a positive measure set, then $H_{\omega}$ has homogeneous ULE in $\mathrm{supp}(\nu)$. 
Also if $T$ is minimal, and $U_{\omega}$ has ULE at a single $\omega$, then $H_{\omega}$ has homogeneous ULE in $\Omega$.
\end{theorem}
For homogeneous ULE $H_{\omega}$, a direct computation as in \cite{DJLS} shows that
\begin{align*}
\sup_{s,t\in I} |\la \delta_x, \rho(s,t,H_{\omega}) \delta_y\ra|\leq C e^{-\mu' |x-y|},
\end{align*}
holds for some $\omega$ independent parameter $0<\mu'<\mu$ and any $\omega\in \mathrm{supp}(\nu)$. This provides the point-wise assumption \eqref{eq:2point_pw} that is required in Theorem \ref{thm:determinant}.
Therefore, we have
\begin{align*}
\sup_{s,t\in I^n} \left|\det( \la \delta_{x_j}, \rho(s_j, t_k, H_{\omega}) \delta_{y_k}\ra )\right|\leq C_{\mu',d}^n e^{-\mu' D_s(X,Y)},
\end{align*}
for any $\omega\in \mathrm{supp}(\nu)$, and some $\omega$-independent $C_{\mu',d}, \mu'>0$ and any configurations $X,Y\subset \Z^d$.

\subsection{Proof of Theorem \ref{thm:loc}}
First, in order to apply our Theorem \ref{thm:determinant}, we need to convert the two-point dynamical localization in expectation assumption \eqref{eq:2point_E} into a pointwise condition (see \eqref{eq:2point_P} below). This reduction goes back to Theorem 7.6 of \cite{DJLS}. 
We include it here for completeness.
Let  
\begin{align*}
Q(\omega):=\sum_{x, y\in \Z^d} (1+|x|)^{-d-1} e^{\frac{\mu}{2}|x-y|} \sup_{s,t\in I} |\la \delta_x, \rho(s,t, H_{\omega}) \delta_y\ra|.
\end{align*}
By \eqref{eq:2point_E}, we have
\begin{align*}
\E_{\omega}\, [Q(\omega)]\leq \sum_{x,y\in \Z^d} (1+|x|)^{-d-1} e^{-\frac{\mu}{2}|x-y|}<\infty,
\end{align*}
which implies $Q(\omega)<\infty$ for a.e. $\omega$. 
Hence for a.e. $\omega$, there exists $C_{\omega}>0$ such that 
\begin{align*}
\sum_{x, y\in \Z^d} (1+|x|)^{-d-1} e^{\frac{\mu}{2}|x-y|} \sup_{s,t\in I} |\la \delta_x, \rho(s,t, H_{\omega}) \delta_y\ra|<C_{\omega},
\end{align*}
which implies for any $x,y\in \Z^d$,
\begin{align}\label{eq:2point_P}
 \sup_{s,t\in I} |\la \delta_x, \rho(s,t, H_{\omega}) \delta_y\ra|<C_{\omega} (1+|x|)^{d+1} e^{-\frac{\mu}{2}|x-y|}.
\end{align}
Theorem \ref{thm:determinant} yields that
\begin{align*}
\sup_{s,t\in I^n} |\det(\la \delta_{x_j}, \rho(s_j, t_k, H_{\omega}) \delta_{y_k}\ra)| 
\leq C_{\omega,d,\mu}^n \left(\prod_{j=1}^n (1+|x_j|)^{d+1}\right)\, e^{-\frac{\mu}{4\sqrt{d}} D_s(X,Y)}
\end{align*}
as claimed.
\qed 

\section*{Acknowledgement}
R. H. is partially supported by NSF-DMS-2053285.
F. Y. is partially supported by an AMS-Simons Travel Grant. 
We would like to thank Michael Loss, Robert Sims and Simone Warzel for helpful comments, and Michael Loss for pointing out that an earlier version of our Theorem \ref{thm:determinant} also holds for permanents of matrices.

\bibliographystyle{amsplain}

\begin{thebibliography}{10}

\bibitem{ANSS}Abdul‐Rahman, H., Nachtergaele, B., Sims, R. and Stolz, G., 2017. Localization properties of the disordered XY spin chain: A review of mathematical results with an eye toward many‐body localization. Annalen der Physik, 529(7), p.1600280.

\bibitem{A81}Aizenman, M., 1981. Proof of the triviality of $\phi_d^4$ field theory and some mean-field features of Ising models for $d>4$. Physical Review Letters, 47(1), p.1.

\bibitem{A82}Aizenman, M., 1982. Geometric analysis of $\phi^4$ fields and Ising models. Parts I and II. Communications in mathematical Physics, 86(1), pp.1-48.

\bibitem{A}Aizenman, M., 1994. Localization at weak disorder: some elementary bounds. In The state of matter: a volume dedicated to EH Lieb (pp. 367-395).

\bibitem{ABF}Aizenman, M., Barsky, D.J. and Fernández, R., 1987. The phase transition in a general class of Ising-type models is sharp. Journal of Statistical Physics, 47(3), pp.343-374.

\bibitem{ADC}Aizenman, M. and Duminil-Copin, H., 2021. Marginal triviality of the scaling limits of critical 4D Ising and $\phi_4^4$ models. Annals of Mathematics, 194(1), pp.163-235.

\bibitem{AG}Aizenman, M. and Graf, G.M., 1998. Localization bounds for an electron gas. Journal of Physics A: Mathematical and General, 31(32), p.6783.

\bibitem{AM}Aizenman, M. and Molchanov, S., 1993. Localization at large disorder and at extreme energies: An elementary derivations. Communications in Mathematical Physics, 157(2), pp.245-278.

\bibitem{ASF}Aizenman, M., Schenker, J.H., Friedrich, R.M. and Hundertmark, D., 2001. Finite-Volume Fractional-Moment Criteria for Anderson Localization. Communications in Mathematical Physics, 224(1), pp.219-253.

\bibitem{AW}Aizenman, M. and Warzel, S., 2009. Localization bounds for multiparticle systems. Communications in Mathematical Physics, 290(3), pp.903-934.

\bibitem{ABP}Aza, N.J.B., Bru, J.B. and de Siqueira Pedra, W., 2018. Decay of complex-time determinantal and Pfaffian correlation functionals in lattices. Communications in Mathematical Physics, 360(2), pp.715-726.

\bibitem{BLS}Bellissard, J., Lima, R. and Scoppola, E., 1983. Localization inv-dimensional incommensurate structures. Communications in Mathematical Physics, 88(4), pp.465-477.

\bibitem{BJ}Bourgain, J. and Jitomirskaya, S., 2002. Absolutely continuous spectrum for 1D quasiperiodic operators. Inventiones mathematicae, 148(3), pp.453-463.


\bibitem{BK}Bravyi, S. and König, R., 2012. Disorder-assisted error correction in Majorana chains. Communications in Mathematical Physics, 316(3), pp.641-692.

\bibitem{BM}Bringmann, B. and Mendelson, D., 2021, October. An eigensystem approach to Anderson localization for multi-particle systems. In Annales Henri Poincaré (Vol. 22, No. 10, pp. 3255-3290). Springer International Publishing.

\bibitem{CLL}Carlen, E., Lieb, E.H. and Loss, M., 2006. An inequality of Hadamard type for permanents. Methods and Applications of Analysis, 13(1), pp.1-18.

\bibitem{CKM}Carmona, R., Klein, A. and Martinelli, F., 1987. Anderson localization for Bernoulli and other singular potentials. Communications in Mathematical Physics, 108(1), pp.41-66.

\bibitem{CS}Chapman, J. and Stolz, G., 2015, February. Localization for random block operators related to the XY spin chain. In Annales Henri Poincaré (Vol. 16, No. 2, pp. 405-435). Springer Basel.

\bibitem{Chu}Chulaevsky, V., 2014. Uniform Anderson localization, unimodal eigenstates and simple spectra in a class of “haarsh” deterministic potentials. Journal of Functional Analysis, 267(11), pp.4280-4320.

\bibitem{DG}Damanik, D. and Gan, Z., 2011. Limit-periodic Schrödinger operators with uniformly localized eigenfunctions. Journal d'Analyse Mathématique, 115(1), pp.33-49.

\bibitem{DG2}Damanik, D. and Gan, Z., 2013. Limit-periodic Schrödinger operators on Zd: Uniform localization. Journal of Functional Analysis, 265(3), pp.435-448.

\bibitem{DS}Damanik, D. and Stollmann, P., 2001. Multi-scale analysis implies strong dynamical localization. Geometric and Functional Analysis GAFA, 11(1), pp.11-29.

\bibitem{DJLS}del Rio, R., Jitomirskaya, S., Last, Y. and Simon, B., 1996. Operators with singular continuous spectrum, IV. Hausdorff dimensions, rank one perturbations, and localization. In J. d'Analyse Math.

\bibitem{DKS}Delyon, F., Kunz, H. and Souillard, B., 1983. One-dimensional wave equations in disordered media. Journal of Physics A: Mathematical and General, 16(1), p.25.

\bibitem{DCGR}Duminil-Copin, H., Goswami, S. and Raoufi, A., 2020. Exponential decay of truncated correlations for the Ising model in any dimension for all but the critical temperature. Communications in Mathematical Physics, 374(2), pp.891-921.

\bibitem{ESS}Elgart, A., Shamis, M. and Sodin, S., 2014. Localisation for non-monotone Schrödinger operators. Journal of the European Mathematical Society, 16(5), pp.909-924.

\bibitem{GYZ}Ge, L., You, J. and Zhou, Q., 2019. Exponential dynamical localization: Criterion and applications. arXiv preprint arXiv:1901.04258.

\bibitem{GK}Germinet, F. and Klein, A., 2001. Bootstrap Multiscale Analysis and Localization in Random Media. Communications in Mathematical Physics, 222(2), pp.415-448.

\bibitem{GFP}Grempel, D.R., Fishman, S. and Prange, R.E., 1982. Localization in an incommensurate potential: An exactly solvable model. Physical Review Letters, 49(11), p.833.

\bibitem{GBK}Groeneveld, J., Boel, R.J. and Kasteleyn, P.W., 1978. Correlation-function identities for general planar Ising systems. Physica A: Statistical Mechanics and its Applications, 93(1-2), pp.138-154.

\bibitem{HJS}Hamza, E., Joye, A. and Stolz, G., 2009. Dynamical localization for unitary Anderson models. Mathematical Physics, Analysis and Geometry, 12(4), pp.381-444.


\bibitem{uniform}Han, R., 2016. Uniform localization is always uniform. Proceedings of the American Mathematical Society, 144(2), pp.609-612.

\bibitem{HYPf}Han, R. and Yang, F. in preparation.

\bibitem{J}Jitomirskaya, S.Y., 1997. Continuous spectrum and uniform localization for ergodic Schrödinger operators. journal of functional analysis, 145(2), pp.312-322.

\bibitem{JKru}Jitomirskaya, S. and Krüger, H., 2013. Exponential Dynamical Localization for the Almost Mathieu Operator. Communications in Mathematical Physics, 322(3), pp.877-882.

\bibitem{JKmon}Jitomirskaya, S. and Kachkovskiy, I., 2018. All couplings localization for quasiperiodic operators with monotone potentials. Journal of the European Mathematical Society, 21(3), pp.777-795.

\bibitem{Kmon}Kachkovskiy, I., 2019. Localization for quasiperiodic operators with unbounded monotone potentials. Journal of Functional Analysis, 277(10), pp.3467-3490.

\bibitem{KG}Klein, A. and Germinet, F., 2012. A comprehensive proof of localization for continuous Anderson models with singular random potentials. Journal of the European Mathematical Society, 15(1), pp.53-143.

\bibitem{KS}Kunz, H. and Souillard, B., 1980. Sur le spectre des opérateurs aux différences finies aléatoires. Communications in Mathematical Physics, 78(2), pp.201-246.

\bibitem{Nishant}Rangamani, N., 2022, May. Exponential dynamical localization for random word models. In Annales Henri Poincaré (pp. 1-23). Springer International Publishing.

\bibitem{SML}Schultz, T.D., Mattis, D.C. and Lieb, E.H., 1964. Two-dimensional Ising model as a soluble problem of many fermions. Reviews of Modern Physics, 36(3), p.856.

\bibitem{Simon}Simon, B., 1980. Correlation inequalities and the decay of correlations in ferromagnets. Communications in Mathematical Physics, 77(2), pp.111-126.

\bibitem{SW}Sims, R. and Warzel, S., 2016. Decay of determinantal and Pfaffian correlation functionals in one-dimensional lattices. Communications in Mathematical Physics, 347(3), pp.903-931.




\end{thebibliography}

\

\

\address{Rui Han, rhan@lsu.edu\\
Department of Mathematics,
Louisiana State University, Baton Rouge, USA}
\

\

\address{Fan Yang, ffyangmath@gmail.com\\
Department of Mathematics,
Louisiana State University, Baton Rouge, USA}

\end{document}